\documentclass[5p]{elsarticle}

\usepackage{lineno,hyperref}

\usepackage{amssymb,amsthm}
\newtheorem{theorem}{Theorem}
\newtheorem{corollary}{Corollary}
\newtheorem{lemma}{Lemma}
\newtheorem{assumption}{Assumption}
\newtheorem{definition}{Definition}
\newtheorem{remark}{Remark}

\usepackage{amsfonts}
\usepackage{graphicx,subfig}
\usepackage{epstopdf}
\usepackage{tabularx,booktabs}
\usepackage{algorithm2e}

\usepackage[cmex10]{amsmath}

\modulolinenumbers[5]

\journal{Energy}

\begin{document}

\begin{frontmatter}

\title{Power Limit of Crosswind Kites}

\author{Haocheng Li}
\address{Aerospace Engineering Program, Worcester Polytechnic Institute}

\begin{abstract}
In this paper, a generalized power limit of the cross wind kite energy systems is proposed. Based on the passivity property of the aerodynamic force, the available power which can be harvested by a cross wind kite is derived. For the small side slip angle case, an analytic result is calculated. Furthermore, an algorithm to calculated the real time power limit of the cross wind kites is proposed.

\end{abstract}


\end{frontmatter}

\section{Introduction}
Crosswind kite power is an emerging renewable energy technology which uses kites or gliders to generate power in high altitude wind flow. Compared to the conventional wind turbine technology, the crosswind power systems can achieve high-speed crosswind motion, which increase the power density significantly. Additionally, without the support structure, the construction cost of such wind energy system could be reduced compared to conventional towered wind turbines. This could eventually reduce the cost of the wind energy dramatically. 

Based on the power generation modes, the crosswind power can be placed into two categories: the FlyGen systems and the GroundGen systems. FlyGen systems, also referred to as drag mode systems, generate power through on-board devices, such as wind turbine. GroundGen systems, also referred to as lift mode systems, generate power through tether tension and power generator on the ground. Different industrial prototypes have been proposed for the airborne wind harvesting,  

Although  the theoretical power limit of crosswind kites is the most fundamental issue of such systems. In \cite{Loyd}, Miles Loyd first proposed a power limit for the crosswind kite systems. Diehl provide a refined version of the power limit in \cite{AWEbookChap1} by considering the turbine power generation. 
However, both Loyd and Diehl's works are based on kite motion in a two-dimensional wind field and no side force has been taken into account in their analysis.
In this paper, I propose the theoretical limit of the crosswind kite systems in three-dimensional wind field, and the following contributions are made. First, a general power loss of the crosswind kite system is derived based on the passivity analysis of the classical aerodynamic model. Based on the power loss calculation, the available power of a crosswind kite is derived, and a nonlinear optimization is presented to determine the power limit.

\section{Available Power}
In this section, the available power of the cross wind kites will be derived. First, a classical aerodynamic model is presented and  the passivity property of the aerodynamic force is then proven. Physically, the passivity of the aerodynamic force represents the power dissipated by the aerodynamic force. Then available power of the crosswind kite systems is then derived based on this property of the aerodynamic force. The derivation presented in this section serves as a foundation of the power limit calculation of the cross wind kite in next section. 

\subsection{Classical Model of Aerodynamics}
As in classical aerodynamics, \cite{Anderson} the following two reference frames are used to describe the three-dimensional kite motion,
\begin{itemize}
\item Cartesian Earth Frame \textbf C: $\begin{pmatrix}
\mathbf i_C & \mathbf j_C & \mathbf k_C
\end{pmatrix}$
\item Body Frame \textbf B: $\begin{pmatrix}
\mathbf i_B & \mathbf j_B & \mathbf k_B
\end{pmatrix}$
\end{itemize}
In crosswind kite systems, the Cartesian earth frame \textbf C is often assumed to center at anchor point of the tether. The $\mathbf i_C$ axis points to the upstream direction and $\mathbf j_C$ points vertical downwards. The $\mathbf j_C$ axis forms a right-hand coordinate system with $\mathbf i_C$ and $\mathbf k_C$ axes. The body frame \textbf B, which centers at the gravitational center of the kite, follows the North-East-Down conventions. 

In this work, we will denote the rotational transformation matrix from frame \textbf C to \textbf B as $\mathbf L_{BC}$, which can be represented either by Euler angles or quaternion. The the kite and wind velocity measured in frame \textbf C are denoted as $\mathbf V_C$ and $\textbf W$ respectively, then the kite and wind velocity observed in frame \textbf B are given by,
\begin{align}\label{eq1}
\mathbf V_B = \mathbf L_{BC}\mathbf V_C,\hspace{1em} \mathbf W_B = \mathbf L_{BC}\mathbf W_C.
\end{align}
It is well known that the rotational transformation matrices are orthogonal, i.e., it inverse transformation, denoted as $\mathbf L_{CB}$, is its transpose,
\begin{align}\label{eq2}
\mathbf L_{CB} = \mathbf L_{BC}^{-1} = \mathbf L_{BC}^T.
\end{align}
The kite apparent velocity, which is denoted as $\mathbf V_a$, is the key in determining the aerodynamic force acting on kites. Using the notation in equations \eqref{eq1} and \eqref{eq2},
\begin{align}\label{eq3}
\mathbf V_a = \mathbf L_{BC}(\mathbf V_C - \mathbf W_C) = \begin{pmatrix} u_a & v_a & w_a \end{pmatrix},
\end{align}
for convenience, $-\mathbf V_a$ may also be referred to as apparent wind velocity. Using kite apparent velocity, $\mathbf V_a$, the kite apparent attitudes, angle of attack $\alpha$ and sideslip angle $\beta$, are given by 
\begin{align}\label{eq4}
\alpha = \arctan\big(\frac{w_a}{u_a}\big), \hspace{1em} \beta = \arcsin\big(\frac{v_a}{V_a}\big),
\end{align}
where $V_a = \|\mathbf V_a\|$ is also refer as the kite apparent speed. 

It is a common practice, \cite{Anderson}, that the kite apparent velocity in $\mathbf i_B$ is assumed to be positive, which is clearly state as follows,
\begin{assumption}\label{asm1}
The kite apparent velocity along $\mathbf i_B$ direction is positive, i.e. $u_a > 0$.
\end{assumption}
In conventional wind energy system, \cite{Bernard}, Assumption \ref{asm1} simply implies that the wind turbine is assumed always facing to the wind during the power harvesting. Under this assumption, the kite apparent velocity $\mathbf V_a$ can be represented by the apparent attitudes, $\alpha$ and $\beta$, in the following way,
\begin{align}\label{eq5}
\mathbf V_a =  \|\mathbf V_a\|\begin{pmatrix}
\cos\alpha\cos\beta & \sin\beta & \sin\alpha\cos\beta
\end{pmatrix}^T.
\end{align}
Equation \eqref{eq5} states that the kite apparent velocity is a function of kite apparent attitudes.
In classical aerodynamics, kite lift and drag coefficients, $C_L$ and $C_D$, are also functions of apparent attitudes,
\begin{align}\label{eq6}
C_L = C_L(\alpha), \hspace{1em} C_D = C_D(\alpha).
\end{align}
By convention, lift and drag coefficients are defined based on the flow direction, i.e. direction of $V_\infty$ shown in Figure \ref{fig1}. The aerodynamic coefficients along the body axes are given by applying appropriate trigonometric transformation,
\begin{align*}
\begin{pmatrix}
C_x \\ C_z
\end{pmatrix} & = \begin{pmatrix}
1 & 0\\
0 & -1
\end{pmatrix}\begin{pmatrix}
\sin\alpha & -\cos\alpha\\
\cos\alpha & \sin\alpha
\end{pmatrix}\begin{pmatrix}
C_L\\C_D
\end{pmatrix},\\
\mathbf C_B & = \begin{pmatrix}
C_x & C_y & C_z
\end{pmatrix}^T - \begin{pmatrix}
C_t & 0 & 0
\end{pmatrix}^T,
\end{align*}
where $C_t$ is the turbine drag coefficient and $C_y$ is the side force coefficient.
In a more compact form, the kite aerodynamic coefficient in frame \textbf B is given by,
\begin{align}\label{eq7}
\mathbf C_B =\begin{pmatrix}C_L\sin\alpha-C_D\cos\alpha - C_t  \\  C_y\\  -C_L\cos\alpha-C_D\sin\alpha
    \end{pmatrix}.
\end{align}
Therefore, the steady aerodynamic force acting on the kite, measured in frame \textbf B, can be calculated as
\begin{align}\label{eq8}
\mathbf A_B = \frac{1}{2}\rho SV_a^2\mathbf C_B,
\end{align}
where $S$ is the kite area and $\rho$ is the air density. Correspondingly, the steady aerodynamic force acting on the kite, measured in frame \textbf C, is given by
\begin{align}\label{eq9}
\mathbf A_C = \frac{1}{2}\rho SV_a^2\mathbf L_{CB}\mathbf C_B,
\end{align}
To this end, all important notations that will be used in the rest of the paper have been introduced. To achieve a more general power limit of crosswind kite, the passivity of the aerodynamic force should be first discussed. 
\begin{figure}
\begin{centering}
\includegraphics[width = 0.4\textwidth]{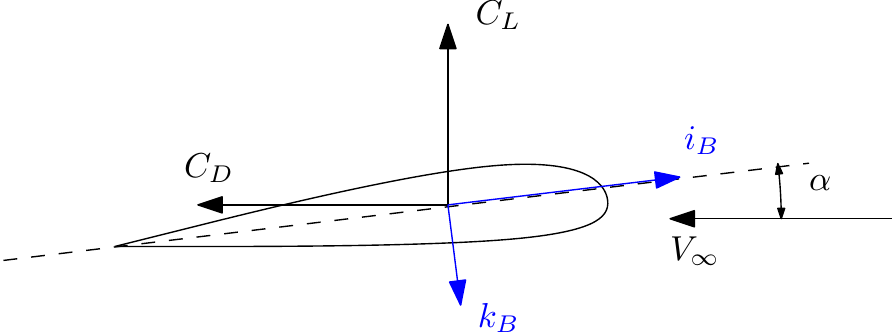}
\caption{Lift and Drag Coefficients}
\label{fig1}
\end{centering}
\end{figure}

\subsection{Power Loss}
To establish the passivity property of the aerodynamic model given in equations \eqref{eq8} and \eqref{eq9}, some definitions in nonlinear system theory, \cite{Khalil}, need to be reviewed first.
Consider a function, $\mathbf y = \mathbf h(\mathbf u)$, where $\mathbf y$ and $\mathbf u$ are input and output signal with compatible dimensions.
Then, the passivity of function $\mathbf h(\mathbf u)$ is defined formally as follows, 
\begin{definition}\label{def1}
A function $\mathbf y = \mathbf h(\mathbf u)$ is Passive if $\mathbf u^T\mathbf y \geq 0$. The function is strictly input passive if there exist a function $\boldsymbol\varphi(\mathbf u)$, with proper dimension, such that $\mathbf u^T\boldsymbol\varphi(\mathbf u) > 0$ for all $\mathbf u \neq \mathbf 0$ and $\mathbf u^T\mathbf y \geq \mathbf u^T\boldsymbol\varphi(\mathbf u)$.
\end{definition}
The classical aerodynamics model introduced in the previous section, i.e  equations \eqref{eq3} - \eqref{eq8}, can be represented using block diagram shown in Figure \ref{fig2}.
\begin{figure}[htb]
\centering
\includegraphics[width = 0.45\textwidth]{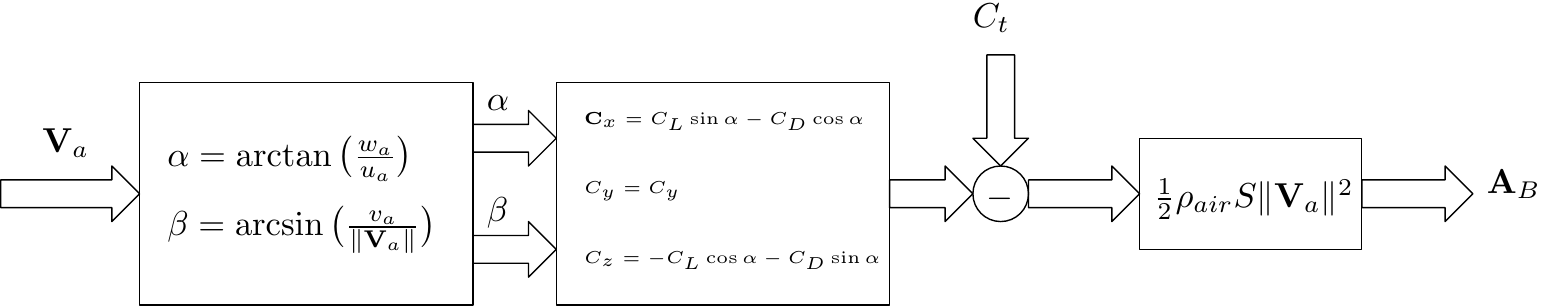}
\caption{Block Diagram of Aerodynamic Model}
\label{fig2}
\end{figure}
The kite apparent velocity, $\mathbf V_a$, can be treated as input to the model while the aerodynamic force in frame \textbf B, $\mathbf A_B$, is the output of the model. Then the following lemma can be proven,
\begin{lemma}\label{lem1}
Under the Assumption \ref{asm1}, the aerodynamic force is strictly input passive with respect to apparent wind velocity. 
\begin{align}\label{eq10}
    &- \mathbf V_a^T\mathbf A_B =  - \frac{1}{2}\rho\|\mathbf{V}_a\|^3SC_a>0.
\end{align}
where $C_a$ is given by,
\begin{align*}
    C_a = -C_D\cos\beta+C_y\sin\beta -C_t\cos\alpha\cos\beta.
\end{align*}
\end{lemma}
Please refer to \cite{ApprntAttudTrck} for detailed proof of this lemma. The passivity of the kite consists of two parts, the power dissipation due to kite structure and the power harvested by the on board turbine. Therefore, 
\begin{remark}\label{rmk1}
The pure power loss due to the kite structure is given by,
\begin{align}\label{eq11}
\mathbf V_a^T\mathbf A_k = \frac{1}{2}\rho\|\mathbf{V}_a\|^3S\big(-C_D\cos\beta+C_y\sin\beta\big),
\end{align}
where $\mathbf A_k = \mathbf A_B$ with $C_t = 0$.
\end{remark}

\subsection{Available Power of Crosswind Kites}
Using the power extraction formula given in \cite{AWEbookChap1}, the total power that can be extracted from the wind in the kite cross wind motion can be computed as follows,  
\begin{align}\label{eq12}
P_{t} = \mathbf W_C^T\mathbf A_C = \mathbf W_B^T\mathbf A_B.
\end{align}
The second equality can be derived by using equations \eqref{eq1}, \eqref{eq2} and \eqref{eq9}. Combining with Remark \ref{rmk1}, the available power in a cross wind motion is
\begin{align}\label{eq13}
    P_a = \mathbf W_B^T\mathbf A_B + \mathbf V_a^T\mathbf A_k.
\end{align}
Denote the wind velocity measured in frame \textbf B as,
\begin{align}\label{eq14}
\mathbf W_B = \begin{pmatrix}
W_x & W_y & W_z
\end{pmatrix}^T.
\end{align}
Then by substituting equations \eqref{eq5}, \eqref{eq7}, \eqref{eq8} and \eqref{eq11} into equation \eqref{eq13}, the available power is given by 
\begin{align}\nonumber
    P_a & = \frac{1}{2}\rho\|\mathbf V_a\|^2S\mathbf W_B^T\mathbf C_B + \mathbf V_a^T\mathbf A_k\\\nonumber
    & = \frac{1}{2} \rho \|\mathbf V_a\|^2 S \Big(W_x(C_L\sin\alpha - C_D\cos\alpha - C_t)  \\\nonumber
    & + W_yC_y - W_z(C_L\cos\alpha + C_D\sin\alpha) \\\label{eq15}
    & + \|\mathbf V_a\|(-C_D\cos\beta + C_y\sin\beta)\Big).
\end{align}
Hence, it is clear that the theoretical power limit of a crosswind kite system can be computed from the following nonlinear optimization problem,
\begin{align}\nonumber
\max_{\alpha, \beta, V_a} \hspace{0.5em} & \frac{1}{2}\rho S\|\mathbf V_a\|^2 \Big(W_x(C_L\sin\alpha - C_D\cos\alpha - C_t) \\\nonumber
& + W_yC_y - W_z(C_L\cos\alpha + C_D\sin\alpha) \\\label{eq16}
& + \|\mathbf V_a\|(-C_D\cos\beta + C_y\sin\beta)\Big).
\end{align}

\section{Theoretical Power Limits}
It worth noting that the available power given in \eqref{eq15} is a nonlinear function of the angle of attack $\alpha$ and side slip angle $\beta$, and the solution to the nonlinear optimization problem \eqref{eq16} is difficult to obtain. However, under certain simplified assumptions, the theoretical power limit of cross wind kite can be computed analytically. In this section, three different theoretical power limits, under three different assumptions, will be derived. First, I will show that the classical Loyd's limit can be derived by assuming the side force acting on the kite is zero and the turbine drag and kite drag force are co-linear. Then, a power limit that only assuming the side force is zero will be derived. Finally, the power limit, under small side slip angle assumption, will also be derived analytically.

\subsection{Loyd's Limit}
In classical aerodynamics, it is reasonable to assume that the side force is negligible when the side slip angle is zero, that is 
\begin{assumption}\label{asm2}
If the side slip angle is zero then the side force coefficient is also zero, i.e. 
\begin{align}\label{eq17}
    C_y = 0 \hspace{1em} \mbox{if} \hspace{1em} \beta = 0
\end{align}
\end{assumption}
\noindent To derive the Loyd's limit, the following assumptions are also necessary,
\begin{assumption}\label{asm3}
The kite drag and turbine drag forces are co-linear.
\end{assumption}
\noindent Under the Assumption \ref{asm3}, the aerodynamic coefficients given in \eqref{eq7} can be simplified as follows,
\begin{align}\label{eq18}
\mathbf C_B =\begin{pmatrix}C_L\sin\alpha - (C_D + C_t)\cos\alpha  \\  C_y\\  -C_L\cos\alpha - (C_D + C_t)\sin\alpha
    \end{pmatrix}.
\end{align}
Therefore, the available power of the kite given in equation \eqref{eq15} can be approximated by
\begin{align}\nonumber
    P_a  = & \frac{1}{2}\rho S\|\mathbf V_a\|^2 \Big(W_x(C_L\sin\alpha - (C_D+C_t)\cos\alpha ) \\\nonumber
& + W_yC_y - W_z(C_L\cos\alpha + (C_D+C_t)\sin\alpha) \\\label{eq19}
& + \|\mathbf V_a\|(-C_D\cos\beta + C_y\sin\beta)\Big).
\end{align}
Under the assumptions \ref{asm2} and \ref{asm3}, the upper bound of the available power can be found as in the following theorem,
\begin{theorem}\label{thm1}
Assume that the sideslip angle $\beta$ is zero and the kite drag is co-linear with the turbine drag, then the available power of the crosswind kite, with on-board turbine, is bounded by 
\begin{align}\label{eq20}
P_a \leq \frac{2}{27}\rho S(\sqrt{W_x^2 + W_z^2})^3\frac{(\sqrt{C_L^2 + (C_D + C_t)^2})^3}{C_D^2}.
\end{align}
Especially, without on-board turbine, the available power is bounded by 
\begin{align}\label{eq21}
P_a \leq \frac{2}{27}\rho S(\sqrt{W_x^2 + W_z^2})^3\frac{(\sqrt{C_L^2 + C_D ^2})^3}{C_D^2}.
\end{align}
\end{theorem}
\begin{proof}
Under the assumption \ref{asm2}, if $\beta = 0$ then $C_y = 0$, and the available power \eqref{eq19} can be simplified as follows,
\begin{align}\nonumber
    P_a  = & \frac{1}{2}\rho S\|\mathbf V_a\|^2 \Big(W_x(C_L\sin\alpha - (C_D+C_t)\cos\alpha ) \\\label{eq22}
  -& W_z(C_L\cos\alpha + (C_D+C_t)\sin\alpha) - \|\mathbf V_a\|C_D \Big).
\end{align}
Notice that
\begin{align}\nonumber
    & W_x(C_L\sin\alpha - (C_D+C_t)\cos\alpha ) \\\label{eq48}
    & - W_z(C_L\cos\alpha + (C_D+C_t)\sin\alpha) = (\mathbf W_B^{'})^T \mathbf C_B^{'},
\end{align}
where $\mathbf W_B^{'}$ and $\mathbf C_B^{'}$ are defined as
\begin{align*}
   \mathbf W_B^{'} & = \begin{pmatrix}
W_x & W_z
\end{pmatrix},\\
\mathbf C_B^{'} & = \begin{pmatrix}
C_L\sin\alpha - (C_D+C_t)\cos\alpha \\ - C_L\cos\alpha - (C_D+C_t)\sin\alpha
\end{pmatrix}.
\end{align*}
It is clear that the norm of $\mathbf C_B^{'}$ and $\mathbf W_B^{'}$ are given by 
\begin{align*}
    \|\mathbf C_B^{'}\| = & \sqrt{C_L^2 + (C_D + C_t)^2}\\
    \|\mathbf W_B^{'}\| = & \sqrt{W_x^2 + W_z^2}
\end{align*}
Therefore the upper bound of the available power $P_a$ is
\begin{align}\nonumber
P_a
\leq & \frac{1}{2}\rho S\|\mathbf V_a\|^2\Big(\sqrt{W_x^2 + W_z^2}\sqrt{C_L^2 + (C_D + C_t)^2}\\ \label{eq24}
& - \|\mathbf V_a\|C_D\Big).
\end{align}
Differentiating the right hand side of \eqref{eq24} with respect to $\|\mathbf V_a\|$ gives, 
\begin{align*}
&\frac{\partial }{\partial \|\mathbf V_a\|}\bigg(\|\mathbf V_a\|^2\Big(\|\mathbf W_B^{'}\|\sqrt{C_L^2 + (C_D + C_t)^2} - \|\mathbf V_a\|C_D\Big)\bigg)\\
 &=  \|\mathbf V_a\|\Big(2\|\mathbf W_B^{'}\|\sqrt{C_L^2 + (C_D + C_t)^2} - 3\|\mathbf V_a\|C_D\Big).
\end{align*}
The maximizing apparent speed can be solved by the following equation,
\begin{align*}
\|\mathbf V_a\|^\ast\Big(2\|\mathbf W_B^{'}\|\sqrt{C_L^2 + (C_D + C_t)^2} - 3\|\mathbf V_a\|^\ast C_D\Big) = 0.
\end{align*}
That is,
\begin{align*}
\|\mathbf V_a\|^\ast = 0 \mbox{ or } \|\mathbf V_a\|^\ast = \frac{2}{3}\|\mathbf W_B^{'}\|\frac{\sqrt{C_L^2 + (C_D + C_t)^2}}{C_D}.
\end{align*}
If $\|\mathbf V_a\|^\ast = 0$, it is clear that the available power $P_a$ is zero, therefore the maximizing apparent speed is 
\begin{align}\label{eq25}
\|\mathbf V_a\|^\ast = \frac{2}{3}\sqrt{W_x^2 + W_z^2}\frac{\sqrt{C_L^2 + (C_D + C_t)^2}}{C_D}.
\end{align}
Substituting equation \eqref{eq25} into \eqref{eq24} gives,
\begin{align}\label{eq26}
P_a \leq \frac{2}{27}\rho S(\sqrt{W_x^2 + W_z^2})^3\frac{(\sqrt{C_L^2 + (C_D + C_t)^2})^3}{C_D^2}.
\end{align}
This is exactly the power limit given by Diehl in \cite{AWEbookChap1}. The Loyd's limit, which is given in \cite{Loyd}, can be obtained by simply setting $C_t = 0$ in inequality \eqref{eq26}.
\end{proof}

It has been shown, in Theorem \ref{thm1}, that Loyd's limit can be derived using the proposed optimization problem under assumptions \ref{asm2} and \ref{asm3}. However, if Assumption \ref{asm3} is relaxed, a different power limit can be achieved.
\begin{corollary}\label{col1}
Assume the sideslip angle $\beta = 0$, then the power limit of the crosswind kite system is given by 
\begin{align}\label{eq27}
    P_a \leq \frac{2}{27}\rho S\frac{(\sqrt{(W_x^2 + W_z^2)(C_L^2 + C_D^2)}-W_xC_t)^3}{C_D^2}.
\end{align}
\end{corollary}
\begin{proof}
If the side-slip angle $\beta = 0$, then the side force coefficient $C_y = 0$. The available power $P_a$, given in equation \eqref{eq15}, can be simplified as follows,
\begin{align}\nonumber
    P_a
    & = \frac{1}{2} \rho \|\mathbf V_a\|^2 S \Big(W_x(C_L\sin\alpha - C_D\cos\alpha - C_t)  \\\label{eq28}
    & - W_z(C_L\cos\alpha + C_D\sin\alpha) - \|\mathbf V_a\|C_D\Big).
\end{align}

The upper bound of $P_a$ over all $\alpha$ can be found using the following inequality,
\begin{align}\nonumber
& W_x(C_L\sin\alpha - C_D\cos\alpha) -  W_z(C_L\cos\alpha + C_D\sin\alpha)\\\nonumber
= & \begin{pmatrix}
W_x & W_y
\end{pmatrix}\begin{pmatrix}
C_L\sin\alpha - C_D\cos\alpha \\ -C_L\cos\alpha - C_D\sin\alpha
\end{pmatrix}\\\label{eq29}
\leq & \sqrt{(W_x^2 + W_z^2)(C_L^2 + C_D^2)}.
\end{align}
Substituting inequality \eqref{eq29} into \eqref{eq28} gives 
\begin{align}\nonumber
P_a & \leq \frac{1}{2} \rho \|\mathbf V_a\|^2 S \Big(\sqrt{(W_x^2 + W_z^2)(C_L^2 + C_D^2)}-W_xC_t  \\\label{eq30}
    &   - \|\mathbf V_a\|C_D \Big).
\end{align}
Taking the derivative of the right hand side of equation \eqref{eq30} with respect to $\|\mathbf V_a\|$ gives that
\begin{align*}
     2\Big(\sqrt{(W_x^2 + W_z^2)(C_L^2 + C_D^2)}-W_xC_t \Big) - 3\|\mathbf V_a\|^\ast C_D = 0.
\end{align*}
That is
\begin{align}\label{eq31}
    \|\mathbf V_a\|^\ast = \frac{2}{3}\frac{\sqrt{(W_x^2 + W_z^2)(C_L^2 + C_D^2)}-W_xC_t}{C_D}.
\end{align}
Therefore, the upper bound is given by substituting equation \eqref{eq31} into equation \eqref{eq30},
\begin{align*}
    P_a \leq \frac{2}{27}\rho S\frac{(\sqrt{(W_x^2 + W_z^2)(C_L^2 + C_D^2)}-W_xC_t)^3}{C_D^2}.
\end{align*}
\end{proof}

\subsection{Small Side Slip Case}
Although the optimization problem \eqref{eq16} is difficult to solve directly, for the case in which the side slip angle is small, an analytic result can be found. 
\begin{theorem}\label{thm2}
If the side slip angle is small such that the side force coefficient $C_y$ can be approximated by the linear function,
\begin{align}\label{eq32}
    C_y \approx C_\beta\beta,
\end{align}
then, the power limit of the crosswind kite system is given by 
\begin{align}\label{eq33}
P_a & \leq \frac{1}{2} \rho S \frac{\big(\gamma_1^2 + 2C_D^2\gamma_2^2 + \gamma_1\gamma_3\big)(\gamma_1 + \gamma_3)}{27C_D^2},\\\nonumber
\gamma_1 & = \sqrt{(W_x^2 + W_z^2)(C_L^2 + C_D^2)} -W_xC_t,\\\nonumber
\gamma_2 & = \sqrt{-\frac{3C_\beta}{4C_D}W_y^2},\\\nonumber
\gamma_3 & = \sqrt{\gamma_1^2+C_D^2\gamma_2^2}.
\end{align}
\end{theorem}
\begin{proof}
If the side slip angle is small such that $\sin\beta \approx \beta$, $\cos\beta\approx 1$ and $C_y \approx C_\beta\beta$, then from the second term of inequality \eqref{eq11},
\begin{align*}
    C_y\sin\beta \approx C_\beta\beta^2 \leq 0.
\end{align*}
This implies that the coefficient $C_\beta$ is not greater zero, i.e.
\begin{align*}
C_\beta \leq 0.
\end{align*}
Then, the available power equation \eqref{eq15} can then be simplified as follows,
\begin{align}\nonumber
    P_a & = \frac{1}{2}\rho\|\mathbf V_a\|^2S\mathbf W_B^T\mathbf C_B + \mathbf V_a^T\mathbf A_k\\\nonumber
    & = \frac{1}{2} \rho \|\mathbf V_a\|^2 S \Big(W_x(C_L\sin\alpha - C_D\cos\alpha - C_t) + W_yC_y   \\\label{eq34}
    & - W_z(C_L\cos\alpha + C_D\sin\alpha) + \|\mathbf V_a\|(-C_D + C_\beta\beta^2)\Big).
\end{align}
Similar to equations \eqref{eq28} - \eqref{eq30}, the available power \eqref{eq34} is bounded above by,
\begin{align}\nonumber
P_a & \leq \frac{1}{2} \rho \|\mathbf V_a\|^2 S \Big(\sqrt{(W_x^2 + W_z^2)(C_L^2 + C_D^2)}-W_xC_t  \\\label{eq35}
    &  + W_yC_\beta\beta + \|\mathbf V_a\|(-C_D +C_\beta\beta^2)\Big).
\end{align}

Notice that the right hand side of $\eqref{eq35}$ is concave with respect to $\beta$ since $C_\beta$ is not positive. The maximizing side slip angle is determined by 
\begin{align*}
2\|\mathbf V_a\|C_\beta\beta^\ast + W_yC_\beta = 0,
\end{align*}
which implies
\begin{align}\label{eq36}
\beta^\ast = -\frac{W_y}{2\|\mathbf V_a\|}.
\end{align}
Substituting equation \eqref{eq36} into the right hand side of \eqref{eq35} yields,
\begin{align}\nonumber
P_a \leq & \frac{1}{2} \rho \|\mathbf V_a\|^2 S(-\|\mathbf V_a\|C_D + \sqrt{(W_x^2 + W_z^2)(C_L^2 + C_D^2)}\\\label{eq37}
& -W_xC_t - C_\beta\frac{W_y^2}{4\|\mathbf V_a\|}).
\end{align}

By $C_\beta \leq 0$, we have $\gamma_2 \geq 0$. The inequality \eqref{eq37} then can be simplified using notation $\gamma_1$ and $\gamma_2$,
\begin{align}\label{eq38} 
P_a \leq & \frac{1}{2} \rho \|\mathbf V_a\|^2 S(-\|\mathbf V_a\|C_D + \gamma_1 +\frac{C_D}{3}\frac{\gamma_2^2}{\|\mathbf V_a\|}).
\end{align}
Taking the derivative of the right hand side of \eqref{eq38} with respect to $\|\mathbf V_a\|$ gives
\begin{align}\nonumber
& \frac{\partial}{\partial \|\mathbf V_a\|}\|\mathbf V_a\|^2 (-\|\mathbf V_a\|C_D + \gamma_1 +\frac{C_D}{3}\frac{\gamma_2^2}{\|\mathbf V_a\|})\\\nonumber
= & -3C_D\|\mathbf V_a\|^2 + 2\gamma_1\|\mathbf V_a\| + \frac{C_D\gamma_2^2}{3}.
\end{align}
The optimizing apparent speed $\|\mathbf V_a\|^\ast$ satisfies the following equation,
\begin{align}\label{eq39}
-3C_D(\|\mathbf V_a\|^\ast)^2 + 2\gamma_1\|\mathbf V_a\|^\ast + \frac{C_D\gamma_2^2}{3} = 0.
\end{align}
whose roots are given by 
\begin{align*}
\|\mathbf V_a\|^\ast = \frac{\gamma_1 \pm \sqrt{\gamma_1^2+C_D^2\gamma_2^2}}{3C_D}.
\end{align*}
Clearly, the positive root should be chosen, i.e.
\begin{align}\label{eq40}
\|\mathbf V_a\|^\ast = \frac{\gamma_1 + \sqrt{\gamma_1^2+C_D^2\gamma_2^2}}{3C_D}.
\end{align}
Substituting the optimizing apparent speed \eqref{eq40} into the right hand side of inequality \eqref{eq35} gives,
\begin{align*}
P_a & \leq \frac{1}{2} \rho S \frac{\big(\gamma_1^2 + 2C_D^2\gamma_2^2 + \gamma_1\gamma_3\big)(\gamma_1 + \gamma_3)}{27C_D^2},\\\nonumber
\gamma_1 & = \sqrt{(W_x^2 + W_z^2)(C_L^2 + C_D^2)} -W_xC_t,\\\nonumber
\gamma_2 & = \sqrt{-\frac{3C_\beta}{4C_D}W_y^2},\\\nonumber
\gamma_3 & = \sqrt{\gamma_1^2+C_D^2\gamma_2^2}.
\end{align*}
\end{proof}

\section{Relations Between Power Limits}
To this end, different power limits have been obtained from the optimization problem \eqref{eq16} using different set of assumptions. 
In this section, the relation between power limits \eqref{eq20}, \eqref{eq21}, \eqref{eq27} and \eqref{eq33} will be demonstrated. For simplicity, let us take the following notations,
\begin{align}\label{eq41}
    \mathcal P_1 & =\frac{2}{27}\rho S\big(\sqrt{W_x^2 + W_z^2}\big)^3\frac{(\sqrt{C_L^2 + C_D ^2})^3}{C_D^2},\\\label{eq42}
    \mathcal P_2 & = \frac{2}{27}\rho S\big(\sqrt{W_x^2 + W_z^2}\big)^3\frac{(\sqrt{C_L^2 + (C_D + C_t)^2})^3}{C_D^2},\\\label{eq43}
    \mathcal P_3 & = \frac{2}{27}\rho S\frac{(\sqrt{(W_x^2 + W_z^2)(C_L^2 + C_D^2)}-W_xC_t)^3}{C_D^2},\\\label{eq44}
    \mathcal P_4 & = \frac{1}{2} \rho S \frac{\big(\gamma_1^2 + 2C_D^2\gamma_2^2 + \gamma_1\gamma_3\big)(\gamma_1 + \gamma_3)}{27C_D^2}.
\end{align}
It will be first shown that if the turbine drag and side forces are negligible then the power limits given in equations \eqref{eq41} - \eqref{eq44} are equivalent. Second, the order relations between these limits will be discussed. 

\begin{remark}\label{rmk2}
If $C_\beta = 0$, then $\mathcal P_3 = \mathcal P_4$. Moreover, if $C_t = 0$ and $C_\beta = 0$, then $\mathcal P_1 = \mathcal P_2 = \mathcal P_3 = \mathcal P_4$.
\end{remark}
\begin{proof}
Using definitions \eqref{eq41}-\eqref{eq43}, it is clear that if $C_t = 0$, then $\mathcal P_1 = \mathcal P_2 = \mathcal P_3$.
On the other hand, by definition \eqref{eq33},
\begin{align}\label{eq45}
   \gamma_2 = 0 \text{ if } C_\beta = 0.
\end{align}
For a power generation operation, it is reasonable to assume that $W_x \leq 0$ as shown in Figure \ref{fig1}. This indicates that 
\begin{align}\label{eq46}
\gamma_1 \geq 0.
\end{align}
Therefore, if $\gamma_2 = 0$ then \begin{align}\label{eq47}
   \gamma_3 = \gamma_1.
\end{align}
using equations \eqref{eq45} and \eqref{eq47}, equation \eqref{eq33} can then be simplified as follows,
\begin{align*}
    & \frac{1}{2} \rho S \frac{\big(\gamma_1^2 + 2C_D^2\gamma_2^2 + \gamma_1\gamma_3\big)(\gamma_1 + \gamma_3)}{27C_D^2}\\
    = & \frac{1}{2} \rho S \frac{\big(\gamma_1^2 + \gamma_1\gamma_3\big)(\gamma_1 + \gamma_3)}{27C_D^2}\\
    = & \frac{2}{27}\frac{\rho S\gamma_1^3}{C_D^2}
\end{align*}
Using the definition of $\gamma_1$, it is clear that if $C_\beta = 0$, then $\mathcal P_3 = \mathcal P_4$. Hence if $C_t = 0$ and $C_\beta = 0$, then $\mathcal P_1 = \mathcal P_2 = \mathcal P_3 = \mathcal P_4$.
\end{proof}
The above remark shows that the power limits derived in the previous section reduce to Loyd's limit if the turbine drag and side force are negligible. 

\begin{remark}\label{rmk3}
The order relations between the power limits \eqref{eq41} to \eqref{eq44} is 
$\mathcal P_1 \leq \mathcal P_2, \mathcal P_1 \leq \mathcal P_3 \leq \mathcal P_4$
Moreover, if $W_x = 0$, then $\mathcal P_3 \leq \mathcal P_2$. On the other hand, if $W_z = 0$, then $\mathcal P_2 \leq \mathcal P_3$.
\end{remark}
\begin{proof}
It is clear that $\mathcal P_1 \leq \mathcal P_2$ and $\mathcal P_1 \leq \mathcal P_3$ by definition. It is also clear that $\gamma_3$ is increasing with respect to $\gamma_2^2$. For a power generation operation, it is reasonable to assume that $W_x < 0$, which implies that the power harvest by the turbine $-W_xC_t > 0$. This also implies that $\gamma_1 > 0$. Using definition \eqref{eq44}, $\mathcal P_4$ is a therefore a strictly increasing function of $\gamma_2^2$. Hence,
\begin{align}\label{eq48}
    \min_{\gamma_2}\mathcal P_4 = \mathcal P_4|_{\gamma_2 = 0} = \mathcal P_3.
\end{align}
The second equality holds according to Remark \ref{rmk2}. By the above derivation, we have proven 
\begin{align*}
    \mathcal P_1 \leq \mathcal P_3 \leq \mathcal P_4.
\end{align*}

If $W_x = 0$, then $\mathcal P_2$ and $\mathcal P_3$ can be simplified as follows,
\begin{align}\label{eq49}
    \mathcal P_2|_{W_x = 0} & = \frac{2}{27}\rho S|W_z|^3\frac{(\sqrt{C_L^2 + (C_D + C_t)^2})^3}{C_D^2},\\\label{eq50}
    \mathcal P_3|_{W_x = 0} & = \frac{2}{27}\rho S|W_z|^3\frac{(\sqrt{C_L^2 + C_D^2})^3}{C_D^2}.
\end{align}
Therefore, it is clear that $\mathcal P_3|_{W_x = 0}  \leq \mathcal P_2|_{W_x = 0}$ by equations \eqref{eq49} and \eqref{eq50}. On the other hand, if $W_z = 0$ then $\mathcal P_2$ and $\mathcal P_3$ can be simplified as follows,
\begin{align}\label{eq51}
    \mathcal P_2|_{W_z = 0} & = \frac{2}{27}\rho S|W_x|^3\frac{(\sqrt{C_L^2 + (C_D + C_t)^2})^3}{C_D^2},\\\label{eq52}
    \mathcal P_3|_{W_z = 0} & = \frac{2}{27}\rho S\frac{(|W_x|\sqrt{(C_L^2 + C_D^2)}-W_xC_t)^3}{C_D^2}.
\end{align}

For a turbine generation operation, $W_x \leq 0$, that is, equation \eqref{eq52} can be further simplified as below,
\begin{align}\label{eq53}
    \mathcal P_3|_{W_z = 0} & = \frac{2}{27}\rho S|W_x|^3\frac{(\sqrt{(C_L^2 + C_D^2)} + C_t)^3}{C_D^2}.
\end{align}
Calculate the following square difference,
\begin{align*}
    &(\sqrt{(C_L^2 + C_D^2)} + C_t)^2 -\big( C_L^2 + (C_D + C_t)^2\big)\\
 = & 2C_t\sqrt{(C_L^2 + C_D^2)} - 2C_DC_t \geq 0
\end{align*}
It is clear that $\mathcal P_2|_{W_z = 0} \leq \mathcal P_3|_{W_z = 0}$.
\end{proof}

$\mathcal P_2$ indicates that the power limit is only determined by the magnitude of the wind velocity if $C_L$ and $C_D$ are fixed. on the other hand, $\mathcal P_3$ indicates that the wind direction also has influence on the power limit of the cross wind kites. As shown in Remark \ref{rmk3}, if $W_x = 0$, then the turbine coefficient $C_t$ has no contribution to the power limit of crosswind kites. For a cross wind kite with on-board turbine, one could argue that limit $\mathcal P_3$ is more reasonable than $\mathcal P_2$ from this observation, since wind turbine can not harvest any power from a wind perpendicular to its direction. Moreover, $\mathcal P_4$ indicates that when the side force is considered, the theoretical power limit of the cross wind kites is even higher. This is simply because nonzero side slip angle introduces nonzero side force which increases the total aerodynamic force, hereby increases the total power as defined in equation \eqref{eq12}.

\section{Real Time Power Limit}
In the previous sections, the improved power limits of the cross wind kites have been proposed, and the relations between the proposed limits and Loyd's limit have been discussed. One of the fundamental observation in this paper is that the power limit of cross wind kites is time varying with respect to wind velocity and kite aerodynamic states such as angle of attack and side slip angle. The following simple but nontrivial case implies that under certain situation, the theoretical power limit of the cross wind kites can be lower than the limits given in equation \eqref{eq41}, \eqref{eq43} and \eqref{eq44}.

\begin{corollary}
If $\alpha = \beta = 0$, then the power limit of the cross wind kite is given by 
\begin{align}\label{eq54}
     \mathcal P_0 = & \frac{2}{27}\rho S\frac{(-W_x(C_D + C_t) - W_zC_L)^3}{C_D^2}.
\end{align}
Moreover, the power limit $\mathcal P_0$ satisfies the following relation,
\begin{align}\label{eq55}
    \mathcal P_0 \leq \mathcal P_1.
\end{align}
\end{corollary}
\begin{proof}
If the angle of attack and side slip angle is zero, then available power of cross wind kites is given by 
\begin{align}\nonumber
    P_a(\alpha = 0,\beta = 0) & = \frac{1}{2} \rho \|\mathbf V_a\|^2 S \Big(W_x( - C_D - C_t) - W_zC_L \\ \label{eq56}
    & - \|\mathbf V_a\|C_D\Big)
\end{align}
Taking the derivative of $\mathcal P_a$ with respect to $\|\mathbf V_a\|$ is given by,
\begin{align*}
    \frac{\partial\mathcal P_a}{\partial \|\mathbf V_a\|} = -2\|\mathbf V_a\|(W_x(C_D + C_t) + W_zC_L) - 3\|\mathbf V_a\|^2C_D
\end{align*}
Then the optimum apparent speed of the kite is given by 
\begin{align}\label{eq57}
    \|\mathbf V_a\|^* = -\frac{2}{3}\frac{W_x(C_D + C_t) + W_zC_L}{C_D}.
\end{align}
Substituting equation \eqref{eq57} into \eqref{eq56} gives that 
\begin{align*}
    \mathcal P_0 & = \max_{\|\mathbf V_a\|}\mathcal P_a(\alpha = 0, \beta = 0)\\
    & =  \frac{2}{27}\rho S\frac{(-W_x(C_D + C_t) - W_zC_L)^3}{C_D^2}.
\end{align*}
\end{proof}
It is clear that, if $\alpha = \beta = 0$, then the expression of power limit of cross wind kite is simple. Moreover, it can be shown that limit \eqref{eq54} is also lower than limits given in \eqref{eq41}, \eqref{eq43} and \eqref{eq44}.

\begin{corollary}\label{col3}
If $C_t = 0$, then $\mathcal P_0 \leq \mathcal P_1$. Otherwise, if $C_t \neq 0$, then $\mathcal P_0 \leq \mathcal P_3 \leq \mathcal P_4$.
\end{corollary}
\begin{proof}
From Remark \ref{rmk3}, $\mathcal P_3 \leq \mathcal P_4$.
For conciseness, we shall first prove $\mathcal P_0 \leq \mathcal P_3$. From Corollary \ref{col1}, it is clear that
\begin{align*}
    \mathcal P_0 = \max_{\|\mathbf V_a\|}\mathcal P_a|_{\alpha = 0, \beta = 0} \leq \max_{\|\mathbf V_a\|,\alpha}\mathcal P_a|_{\beta = 0} =  \mathcal P_3
\end{align*}
Hence, if $C_t = 0$, from Remark \ref{rmk2}, the following inequality also holds,
\begin{align*}
     \max_{\|\mathbf V_a\|}\mathcal P_a|_{\alpha = 0, \beta = 0, C_t = 0} \leq \max_{\|\mathbf V_a\|,\alpha}\mathcal P_a|_{\beta = 0, C_t = 0} =  \mathcal P_1.
\end{align*}

\end{proof}

Again, due to the Assumption \ref{asm3}, there is no clear order relation between $\mathcal P_0$ and $\mathcal P_2$ as shown in the following corollary,
\begin{corollary}

\end{corollary}

However, it still suffices to argue that if the angle of attack and side slip angle is known, better estimation on the cross wind kite power limit could be obtained. 
\begin{lemma}\label{lem2}
If $\alpha$ and $\beta$ is known at certain instance $t$, then the power limit of cross wind kite at instance $t$ is given by
\begin{align}\label{eq58}
    P_a \leq \frac{2}{27}\rho S \frac{\bar W^3}{\bar C^2}
\end{align}
where $\bar W$ and $\bar C$ are given by,
\begin{align*}
\bar W = & W_x(C_L\sin\alpha - (C_D+C_t)\cos\alpha ) \\
& + W_yC_y - W_z(C_L\cos\alpha + (C_D+C_t)\sin\alpha), \\
\bar C = & -C_D\cos\beta + C_y\sin\beta.
\end{align*}
\end{lemma}
Lemma \ref{lem2} can be proven similar to equation \eqref{eq24} - \eqref{eq26}, which will not be repeated here. Notice that in steady aerodynamics, the lift and drag coefficients are functions of angle of attack $\alpha$ and the side force coefficient is function of side slip angle $\beta$. Therefore, given $\alpha$ and $\beta$, $\bar W$ and $\bar C$ are also can be calculated. Therefore, in principle, Lemma \ref{lem2} gives us an approach to calculate the real time power limit of a cross wind kite using the following three steps,
\begin{enumerate}
    \item Measure the angle of attack and side slip angle of the cross wind kite.
    \item Calculate $\bar W$ and $\bar C$ using definitions.
    \item Calculate the power limit using equation \eqref{eq58}. 
\end{enumerate}

\section{Conclusion}
In this paper, the modification of the power limit formula of the crosswind kite energy systems is discussed. First, a classical aerodynamic model was presented, and its passivity was derived. Based on this passivity property, the power loss of the cross wind kite systems was obtained. A nonlinear optimization formulation, which determines the available power of the crosswind kites, was proposed. Then the classical Loyd's limit is derived using this formulation. Moreover, a generalization of the power limit formula is also derived. Then the order relations between the proposed limits and classical Loyd's limit are proven. Finally, a real time power limit calculation method is provided. 

\bibliography{RenewableEnergy}

\end{document}